\documentclass[a4paper]{article}

\usepackage[UKenglish]{babel}
\usepackage{hyperref}

\usepackage{amsmath}
\usepackage{amssymb}
\usepackage{amsthm}
\usepackage{thmtools}
\usepackage{enumitem} 

\usepackage{algorithm}
\usepackage[noend]{algpseudocode}
\usepackage{xspace}
\usepackage{booktabs}
\usepackage[most]{tcolorbox}
\usepackage{tikz}
\usetikzlibrary{arrows.meta,positioning}
\usepackage{siunitx}
\usepackage{subcaption}

\newtheorem{theorem}{Theorem}

\newtheorem{corollary}[theorem]{Corollary}
\newtheorem{definition}[theorem]{Definition}

\usepackage{cleveref}

\addto\extrasUKenglish{%
}

\crefname{section}{Sect.}{Sect.}
\Crefname{section}{Sect.}{Sect.}
\crefname{subsection}{Sect.}{Sect.}
\Crefname{subsection}{Sect.}{Sect.}
\crefname{subsubsection}{Sect.}{Sect.}
\Crefname{subsubsection}{Sect.}{Sect.}
\crefname{paragraph}{Sect.}{Sect.}
\Crefname{paragraph}{Sect.}{Sect.}
\crefname{subparagraph}{Sect.}{Sect.}
\Crefname{subparagraph}{Sect.}{Sect.}

\crefname{figure}{Fig.}{Figs.}
\Crefname{figure}{Fig.}{Figs.}

\crefname{theorem}{Thm.}{Thms.}
\Crefname{theorem}{Thm.}{Thms.}
\crefname{lemma}{Lem.}{Lems.}
\Crefname{lemma}{Lem.}{Lems.}
\crefname{definition}{Def.}{Defs.}
\Crefname{definition}{Def.}{Defs.}
\crefname{corollary}{Cor.}{Cors.}
\Crefname{corollary}{Cor.}{Cors.}
\crefname{invariant}{Inv.}{Invs.}
\Crefname{invariant}{Inv.}{Invs.}

\crefname{algorithm}{Alg.}{Algs.}
\Crefname{algorithm}{Alg.}{Algs.}

\newcounter{bvechange}
\newcommand{\bvechange}{\stepcounter{bvechange}\textbf{(C\thebvechange)}}

\newcommand{\asg}{:=}
\newcommand{\inprocess}{\texttt{Inprocess}\xspace}
\newcommand{\backtrack}{\texttt{Backtrack}\xspace}

\newcommand{\ifb}{\textbf{if}\;}
\newcommand{\thenb}{\textbf{then}\;}

\newcommand{\csubs}[3]{\ensuremath{#2 \subseteq^{#1} #3}\xspace}
\newcommand{\notsubs}[3]{\ensuremath{#2 \not\subseteq^{#1} #3}\xspace}
\newcommand{\csubsmin}[3]{\ensuremath{#2 \subseteq^{\min(#1)} #3}\xspace}

\newcommand{\al}{\textsc{AL}\xspace}
\newcommand{\gl}{\textsc{BASE}\xspace}
\newcommand{\combined}{\textsc{COMB}\xspace}
\newcommand{\dynamic}{\textsc{DYN}\xspace}

\title{Backtrackable Inprocessing\thanks{Extended version of a paper
accepted to SAT 2026. This version corresponds to the conference
version and additionally includes the appendix, which does not
appear in the proceedings.}}

\author{Alexander Nadel\\[0.5ex]
Faculty of Data and Decision Sciences, Technion, Haifa, Israel\\
\texttt{alexandernad@technion.ac.il}}

\date{} 
\begin{document}

\maketitle

\begin{abstract}
We introduce \emph{Backtrackable Inprocessing} (BI), a framework that enables
applying inprocessing under the current trail at any decision level, at any
point during incremental SAT solving. Our approach lifts the long-standing
restriction that inprocessing must be performed only at the global decision
level, thereby substantially increasing its potential effec\-tiveness. We focus
on three highly efficient core techniques: subsumption, self-subsuming
resolution, and Bounded Variable Elimination (BVE). We show how to ensure
sound backtracking in the presence of inprocessing, and demonstrate that
applying BI for incremental preprocessing after propagating assumptions
yields significant performance improvements on Bounded Model Checking (BMC)
benchmarks from the Hardware Model Checking Competition 2017. Implemented in
the Island SAT solver (IntelSAT's fork), BI enables solving $\sim$1.5$\times$
as many difficult bounds as the baseline global-level incremental
preprocessor.
\end{abstract}

\section{Introduction}
\label{sec:intro}

Modern backtrack-search-based SAT solvers are widely used across numerous applications~\cite{HandbookSAT2021}. A SAT solver determines whether a Boolean formula in Conjunctive Normal Form (CNF) is satisfiable. Since the introduction of MiniSat~\cite{minisat}, \emph{incremental} SAT solving has become essential in many crucial domains~\cite{BacchusJarvisaloMartinsMaxSAT2021,BiereEtAlTACAS99,BradleyIC3VMCAI2011}. In incremental solving, the user performs a sequence of solver queries, each under a set of \emph{assumption literals (assumptions)}~\cite{minisat,NadelRyvchin2012}. Clauses may be added before any query, and each query checks satisfiability of all clauses provided so far under its specific assumptions. Implementations typically realize this by assigning all assumptions as initial decisions (in order) and propagating after each assignment.

\emph{Preprocessing}~\cite{DBLP:conf/sat/EenB05} and, later, \emph{inprocessing}~\cite{DBLP:conf/cade/JarvisaloHB12} have long been known to significantly accelerate SAT solving. Preprocessing simplifies the formula before search, while inprocessing simplifies it regularly during search. One example of an inprocessing technique is \emph{subsumption}: if $\alpha$ and $\beta$ are clauses with $\alpha \subseteq \beta$, then $\alpha$ subsumes $\beta$ and $\beta$ can be removed. Initially restricted to non-incremental solving, inprocessing was adapted for incremental use~\cite{DBLP:conf/sat/NadelRS12,DBLP:conf/sat/FazekasBS19} and has been supported in the mainstream incremental SAT solver CaDiCaL~\cite{DBLP:conf/cav/BiereFFFFP24} since~\cite{DBLP:conf/sat/FazekasBS19}.

However, currently, inprocessing has a crucial limitation: it is applied only \emph{globally at decision level~0}. This restriction significantly limits the power of inprocessing. To illustrate the limitation, consider the clauses $\alpha = a \lor b \lor c$ and $\beta = d \lor b \lor c$. At level~0, $\alpha$ does not subsume $\beta$. But if we inprocess at level~1, after assigning $a \asg 0$, $\alpha$ \emph{does} subsume $\beta$—\emph{until the solver backtracks below level~1}. Removing $\beta$ at this point can be highly beneficial, as it may trigger a chain of further simplifications. This motivates a natural question:

\begin{center}
\emph{Can inprocessing be invoked under the current trail at any decision level?}
\end{center}

No solver offers this capability, with correctness under backtracking the key challenge.

\subsection{Our Contribution: Backtrackable Inprocessing (BI)}
We introduce \emph{Backtrackable Inprocessing (BI)}, the first framework that enables inprocessing under the current trail at any level, at any point in the search. In this work, we focus on three core, extremely effective inprocessing techniques originally applied in the pioneering SatELite preprocessor~\cite{DBLP:conf/sat/EenB05}: subsumption, self-subsuming resolution (which we dub selfsumption), and Bounded Variable Elimination (BVE). BI fully supports incremental SAT solving.

The key challenge is to ensure that the effects of inprocessing are tracked per decision level and can be soundly undone for any level we backtrack over, while remaining in effect at lower levels. We achieve this by generalizing how the solver classifies and manages clauses.

SAT solvers classify clauses into two types: \emph{pervasive} and \emph{temporary}~\cite{DBLP:conf/ftcs/SilvaS97,NadelRyvchin2012}. The set (conjunction) of pervasive clauses is equivalent (or, in some solvers, equisatisfiable) to the set of input clauses, whereas temporaries are implied by the pervasives. Pervasive clauses must be preserved for correctness, while temporary clauses may be deleted at any time (at least at level~0) without affecting correctness and are therefore managed heuristically. Pervasive clauses begin as the input clauses but may change during solving; for example, if a temporary clause $\alpha$ subsumes a pervasive clause $\beta$, then $\beta$ is removed and $\alpha$ becomes pervasive.

We introduce additional clause types for BI, including:
\begin{itemize}
    \item \emph{stashed at level $i$}, which are deactivated until the solver backtracks below~$i$, at which point they are restored;
    \item \emph{deputy at level $i$}, which remain active and cannot be deleted until backtracking below~$i$.
\end{itemize}

As an example of how these clause types are created and  managed,  suppose $\alpha = a \lor b \lor c$ is temporary and $\beta = d \lor b \lor c$ is pervasive under decisions $e \asg 1$ (level~1) and $a \asg 0$ (level~2). After detecting that $\alpha$ subsumes $\beta$ as long as the solver does not backtrack below~2, we mark:
\[
\beta \text{ as a stashed-at-2 clause (deactivated)}, \qquad
\alpha \text{ as a deputy-at-2 clause (active)}.
\]
When the solver backtracks below level~2, $\alpha$ will either be deleted or marked as temporary, and $\beta$ will be restored as pervasive.

In our framework, at all times, each deactivated stashed clause $\beta$ has a deputy clause $\alpha$ that remains active and \emph{represents} $\beta$ (by implying it under the current trail) until either $\beta$ is restored or another stashed clause $\gamma$ that can represent $\beta$ is restored. The correct management of stashed and deputy clauses is the core technical contribution of our work.

To further clarify the role of deputy clauses, consider our example
above but with $\alpha$ pervasive. Clause $\beta$ would be handled as before
(stashed-at-2 and reactivated as pervasive once the solver backtracks below
level~2). The core difference is in the handling of $\alpha$. In the original
example, temporary $\alpha$ must be marked as a deputy to ensure correctness,
by guaranteeing that $\alpha$ is not deleted and continues to represent $\beta$
until $\beta$ is restored. If $\alpha$ is pervasive, it simply remains
so, since pervasive clauses are guaranteed not to be deleted.

We implemented BI in the open-source Island SAT solver, a fork of IntelSAT~\cite{NadelIntelSAT2022}. While, in principle, BI supports any number of inprocessing steps at arbitrary levels, the empirical part of this work focuses on \emph{incremental preprocessing}. We exploit the fact that BI can inprocess at the beginning of each incremental query \emph{after} assumptions are set and propagated, rather than only at level~0 as in standard incremental inprocessing (e.g., in CaDiCaL~\cite{DBLP:conf/cav/BiereFFFFP24,DBLP:conf/sat/FazekasBS19}). We implemented several BI-enabled incremental preprocessing strategies (Sect.~\ref{sec:bistrategies}).

{\sloppy We focus on this setting because incremental preprocessing is the core performance-critical component in incremental SAT-based Bounded Model Checking (BMC). BMC plays a central role in hardware validation, where the goal is to verify the design for as many bounds as possible, and solving each new bound may mean the difference between discovering a crucial bug or missing it entirely. However, each unrolling to a new bound typically results in a substantial increase in the complexity of the SAT instance, making progress toward higher bounds increasingly difficult. Our experiments show that, on BMC benchmarks from the Hardware Model Checking Competition 2017~\cite{BiereVanDijkHeljanko-FMCAD17}, with a one-hour timeout per benchmark and a bound threshold of 100, Island with standard global-level inprocessing solves 640 more bounds than vanilla Island without inprocessing, whereas BI enables solving 948 more bounds, i.e., about 1.5\(\times\) more of the difficult bounds unsolved without inprocessing.\par}

{\sloppy Below, \cref{sec:prelim,sec:bisub,sec:bi} give preliminaries, introduce non‑incremental subsumption‑only BI, and present full BI (proofs in Appendix~\ref{app:proofs}); \cref{sec:exp,sec:conclusion} report experiments and conclude.\par}

\section{Preliminaries}
\label{sec:prelim}

We begin by recalling standard definitions from propositional logic.

\begin{definition}[Variables, literals, clauses, CNF]\label{def:cnf}
Let $\mathcal{V}$ be a countable set of Boolean variables.
A \emph{literal} is either a variable $x \in \mathcal{V}$ or its negation $\lnot x$.
A \emph{clause} is a finite set (disjunction) of pairwise distinct literals.
A CNF formula is a finite set (conjunction) of clauses.
\end{definition}

As is standard in the SAT literature, we represent the SAT solver’s search state by a trail.

\begin{definition}[Trail and decision levels]\label{def:trail}
Let $I$ be an input CNF formula. A \emph{trail} is a finite sequence of literals
annotated with their \emph{decision level}
$\tau \;=\; \langle \ell_1 @ d_1,\; \ldots,\; \ell_k @ d_k \rangle$, such that:
\begin{itemize}
  \item all variables occurring in $\tau$ are pairwise distinct;
  \item the decision levels satisfy $d_1 \leq \cdots \leq d_k$;
  \item for each $i$ with $d_i > 0$, we say that $\ell_i$ is a \emph{decision literal} if there is
        no $j < i$ with $d_j = d_i$; otherwise $\ell_i$ is an \emph{implied literal};
  \item writing $\mathrm{Dec}(\tau)$ and $\mathrm{Impl}(\tau)$ for the sets of decision and implied
        literals on $\tau$, and for $i \in \{1,\ldots,k\}$,
        $\mathrm{Dec}^{< i}(\tau) \;:=\; \{\;\ell_j \in \mathrm{Dec}(\tau) \mid j < i\;\}$,
        we require that for every $i$ with $\ell_i \in \mathrm{Impl}(\tau)$,
        $I \land \bigwedge \mathrm{Dec}^{< i}(\tau) \;\models\; \ell_i$.
\end{itemize}
In other words, in a valid trail every implied literal $\ell_i \in \mathrm{Impl}(\tau)$ is implied by $I$
together with the decision literals appearing before $\ell_i$.
\end{definition}

For the remainder of this paper, we fix the input CNF formula to be $I$
and the trail to be $\tau$; both are assumed to be given whenever
referenced.

As a direct special case of the trail definition at decision level~0, we get that all global level-0 assignments are implied by the input formula; see Cor.~\ref{cor:level0-implied}.

\begin{corollary}[Global literals implied by \(I\)]\label{cor:level0-implied}
Let \(\tau = \langle \ell_1 @ d_1, \ldots, \ell_k @ d_k \rangle\) be a trail.
For every \(i \in \{1,\ldots,k\}\), if \(d_i = 0\), then \(I \models \ell_i\).
\end{corollary}

A trail can be naturally interpreted as a partial assignment consisting of all its literals.
For the example trail $\tau = \langle \lnot a @ 0,\; \lnot b @ 1,\; \lnot c @ 2,\; d @ 2 \rangle$, we have
$\tau \models (a \lor b \lor d)$ and $\tau \not\models (a \lor b \lor \lnot d)$.

We next extend our terminology to define satisfaction at a given decision level.

\begin{definition}[Restriction of a trail to level $d$]\label{def:trail-restrict}
For a trail $\tau = \langle \ell_1 @ d_1,\ldots,\ell_k @ d_k \rangle$ and $d \in \mathbb{N}$, define $\tau^{\le d} \;:=\; \langle \ell_i @ d_i \mid 1 \le i \le k,\ d_i \le d \rangle$, that is, the longest prefix of $\tau$ whose decision levels are at most $d$.
\end{definition}

\begin{definition}[Satisfaction at decision level $d$]\label{def:satisfaction-level-d}
Let $\varphi$ be any Boolean formula and $d \in \mathbb{N}$.
We say that $\varphi$ is \emph{satisfied at level $d$}  if $\tau^{\le d} \models \varphi$.
\end{definition}

We next introduce convenient notation for implication and equivalence evaluated at a given decision level (relative to the fixed input formula $I$ and trail $\tau$).

\begin{definition}[$d$-implication and $d$-equivalence]\label{def:connectives-level-d}
Let $\varphi,\psi$ be Boolean formulas, and $d \in \mathbb{N}$. We say that $\varphi$ \emph{$d$-implies} $\psi$, written $\varphi \to^d \psi$, if $\tau^{\le d} \models (\varphi \to \psi)$, and that $\varphi$ and $\psi$ are \emph{$d$-equivalent}, written $\varphi \leftrightarrow^d \psi$, if $\tau^{\le d} \models (\varphi \leftrightarrow \psi)$.
\end{definition}

For example, given the trail $\tau = \langle e @ 1,\; \lnot a @ 2 \rangle$, we have $(a \lor b \lor c) \to^2 (d \lor b \lor c)$, since $\tau^{\le 2}$ falsifies $a$ and therefore satisfies $(a \lor b \lor c) \to (d \lor b \lor c)$. In contrast, $(a \lor b \lor c) \to^1 (d \lor b \lor c)$ does \textbf{not} hold, as in $\tau^{\le 1}$ only $e$ is assigned and the implication fails.

Cor.~\ref{cor:monotone-d} (level-$d$ entailment persists under trail extension) is immediate from Def.~\ref{def:connectives-level-d}.

\begin{corollary}[Monotonicity of level-$d$ entailment]\label{cor:monotone-d}
Let $\varphi,\psi$ be Boolean formulas, and $i,j \in \mathbb{N}$ with $j \geq i$.
For $\star \in \{\to,\leftrightarrow\}$, if $\varphi \,\star^{\,i} \psi$, then $\varphi \,\star^{\,j} \psi$.
\end{corollary}

For example, with $\tau = \langle e @ 1,\; \lnot a @ 2,\; g @ 3,\; \lnot h @ 4 \rangle$, we have
$(a \lor b \lor c) \to^2 (d \lor b \lor c)$ and hence, by Cor.~\ref{cor:monotone-d}, also
$(a \lor b \lor c) \to^3 (d \lor b \lor c)$ and $(a \lor b \lor c) \to^4 (d \lor b \lor c)$.

We now introduce subsumption and its level-dependent variants  (relative to the trail $\tau$) and relate them to $d$-implication, starting with standard subsumption.

\begin{definition}[Subsumption]\label{def:subsumption}
Clause \(\alpha\) \emph{subsumes} clause \(\beta\), written \(\alpha \subseteq \beta\),
if every literal in \(\alpha\) also occurs in \(\beta\).
\end{definition}

\begin{definition}[Pruning falsified literals]\label{def:pruning}
For a clause \(\alpha\) and \(d \in \mathbb{N}\), we define the
\emph{restriction of \(\alpha\) to level \(d\)} by
\(\alpha^{\le d} := \{\ell \in \alpha \mid \lnot \ell \notin \tau^{\le d}\}\),
that is, we remove from \(\alpha\) all literals whose negation appears on the trail up to level \(d\).
\end{definition}

For the example trail \(\tau = \langle e @ 1, \lnot b @ 2 \rangle\), we have
\((a \lor b \lor c)^{\le 2} = (a \lor c)\), whereas
\((a \lor b \lor c)^{\le 1} = (a \lor b \lor c)\).

\begin{definition}[$d$-subsumption]\label{def:d-subsumption}
Let $\alpha,\beta$ be clauses and $d \in \mathbb{N}$.
We say that $\alpha$ \emph{subsumes $\beta$ at $d$} (\emph{$d$-subsumes} $\beta$),
written $\csubs{d}{\alpha}{\beta}$, if $\alpha^{\le d} \subseteq \beta^{\le d}$.
\end{definition}

For example, under $\tau = \langle \lnot e @ 1,\; \lnot b @ 2 \rangle$, we obtain $\csubs{2}{(a \lor b \lor c)}{(a \lor f \lor c)}$, but $\notsubs{1}{(a \lor b \lor c)}{(a \lor f \lor c)}$ (since $b$ is falsified at level 2 but not at level 1).

It is easy to verify that, as expected, $d$-subsumption entails $d$-implication.

\begin{corollary}[$d$-subsumption entails $d$-implication]\label{cor:csub-impl}
For any clauses $\alpha,\beta$ and any $d \in \mathbb{N}$,  
if $\csubs{d}{\alpha}{\beta}$, then $\alpha \to^d \beta$.
\end{corollary}

We will need the earliest level $d$ at which a subsumption relation holds, so that a subsumed clause can be safely replaced by the subsuming clause until backtracking below $d$.

\begin{definition}[min-$d$-subsumption]\label{def:minimal-d-subsumption}
Let $\alpha,\beta$ be clauses and $d \in \mathbb{N}$.
We say that $\alpha$ \emph{minimally subsumes} $\beta$ at level $d$ ($\alpha$ \emph{min-$d$-subsumes} $\beta$),
written $\csubsmin{d}{\alpha}{\beta}$, if:
\begin{enumerate}
  \item $\csubs{d}{\alpha}{\beta}$ holds, and
  \item for every decision level $d'$ strictly smaller than $d$, 
        $\csubs{d'}{\alpha}{\beta}$ does not hold.
\end{enumerate}
Thus $d$ is the minimal decision level at which $\alpha$ $d$-subsumes $\beta$.
\end{definition}


For example, with \(\tau = \langle e @ 1, \lnot a @ 2, g @ 3, \lnot h @ 4 \rangle\), we have
\((a \lor b \lor c) \to^2 (d \lor b \lor c)\) and hence, by Cor.~\ref{cor:monotone-d}, also
\((a \lor b \lor c) \to^3 (d \lor b \lor c)\) and \((a \lor b \lor c) \to^4 (d \lor b \lor c)\).
However, \((a \lor b \lor c) \not\to^1 (d \lor b \lor c)\), so the implication at level \(2\) is minimal, i.e., $\csubsmin{2}{(a \lor b \lor c)}{(d \lor b \lor c)}$.

\section{Non-Incremental Backtrackable Inprocessing with Subsumption}
\label{sec:bisub}

This section introduces our framework for backtrackable inprocessing and its main algorithms in the following simplified setting: \emph{non-incremental SAT solving with subsumption as the only inprocessing technique}. This restricted setting already captures the main ideas, terminology, invariants, and algorithms. Section~\ref{sec:bi} then lifts these restrictions.

We design our Backtrackable Inprocessing (BI) module to be integrated as a component of a SAT solver with a clear API to the main search engine. The search engine interacts with BI via the following two API functions:
\begin{enumerate}
  \item \inprocess: the main inprocessing function. It can be invoked after propagation has completed following either a decision or a conflict analysis, at an arbitrary decision level, and may be called any number of times during the SAT solver invocation.
 \item \backtrack: let \(m\) be the level at the end of the latest \inprocess or \backtrack invocation. \backtrack \emph{must} be invoked once if the solver backtracks to any level below \(m\).

\end{enumerate}

\subsection{Clause Classification and Management}
Our approach relies on the extended clause classification and management shown in Table~\ref{tab:clauses}. 

\begin{table}[htbp]
\centering
\caption{Clause classification before the first invocation of \inprocess (Table~\ref{tab:clauses-before}) and
after the last \inprocess or \backtrack invocation ending at level~$n$ (Table~\ref{tab:clauses-after}), as well as additional notation for clause groups (Table~\ref{tab:clauses-notation}). The DL column shows each group's decision level. All clause groups $D_i$ and $S_j^q$ are mutually exclusive. First occurrences of symbols are shown on a light gray background.}

\label{tab:clauses}

\begin{subtable}{\textwidth}
  \centering
  \caption{Before invoking \inprocess for the first time (at any decision level)}
  \label{tab:clauses-before}
  \renewcommand{\arraystretch}{1.05} 
  \setlength{\tabcolsep}{4pt}        
  \begin{tabular}{|c|p{0.36\textwidth}|p{0.48\textwidth}|}
    \hline
    \textbf{DL} &
    \textbf{Active Clauses} &
    \textbf{Comment} \\
    \hline
    0 &
    \colorbox{gray!20}{$D_0$}: active-pervasive &
    $D_0$ and $I$ are 0-equivalent, i.e., $D_0 \leftrightarrow^0 I$. \\
    \hline
    $\infty$ &
    \colorbox{gray!20}{$D_\infty$}: temporary &
    Every clause in $D_\infty$ is implied by $D_0$ at level~0, i.e.,
    $D_0 \to^0 D_\infty$. \\
    \hline
  \end{tabular}
\end{subtable}
\vspace{1.0em}

\begin{subtable}{\textwidth}
\centering
\caption{After \inprocess or \backtrack invocation ending at level $n$}
\label{tab:clauses-after}
\begin{tabular}{|c|l|l|}
\hline
\textbf{DL} & \textbf{Active Clauses} & \textbf{Stashed Clauses (Inactive)} \\
\hline
0 & $D_0$: active-pervasive & \\
\hline
1 & \colorbox{gray!20}{$D_1$}: 1-deputy 
  & \colorbox{gray!20}{$S_1$} $= \colorbox{gray!20}{$S_1^0$}$ \\
\hline
2 & \colorbox{gray!20}{$D_2$}: 2-deputy 
  & \colorbox{gray!20}{$S_2$} $= \colorbox{gray!20}{$S_2^0$} \cup \colorbox{gray!20}{$S_2^1$}$ \\
\hline
\vdots & \vdots & \vdots \\
\hline
$n$ & \colorbox{gray!20}{$D_n$}: $n$-deputy 
    & \colorbox{gray!20}{$S_n$} $= \colorbox{gray!20}{$S_n^0$} \cup \colorbox{gray!20}{$S_n^1$} \cup \cdots \cup \colorbox{gray!20}{$S_n^{\,n-1}$}$ \\
\hline
$\infty$ & $D_\infty$: temporary & \\
\hline
\end{tabular}
\end{subtable}

\vspace{1.0em}

\begin{subtable}{\textwidth}
\centering
\caption{Additional notation for clause groups. Note that $S^i \subseteq S$ for every $i$, and $D \subseteq N \subseteq A$.}
\label{tab:clauses-notation}
\begin{tabular}{|c|l|l|c|}
\hline
\textbf{Symbol} & \textbf{Name}            & \textbf{Definition}                          & \textbf{Active?} \\
\hline
\colorbox{gray!20}{$S$}   & stashed              & $S_1 \cup S_2 \cup \cdots \cup S_n$          & No    \\
\hline
\colorbox{gray!20}{$S^i$} & $i$-restored         & $S_1^i \cup S_2^i \cup \cdots \cup S_n^i$    & No    \\
\hline
\colorbox{gray!20}{$P$}   & pervasive            & $D_0 \cup S^0$                               & Mixed \\
\hline
\colorbox{gray!20}{$D$}   & deputy               & $D_1 \cup \cdots \cup D_n$                   & Yes   \\
\hline
\colorbox{gray!20}{$N$}   & active-non-temporary & $D_0 \cup D$                                 & Yes   \\
\hline
\colorbox{gray!20}{$A$}   & active               & $N \cup D_\infty$                            & Yes   \\
\hline\hline
\colorbox{gray!20}{$I$}   & input                & the original input clauses                   &       \\
\hline
\end{tabular}
\end{subtable}

\end{table}

Table~\ref{tab:clauses-before} demonstrates that, before the initial invocation of \inprocess (at whatever decision level), the solver maintains two standard clause types: the active-pervasive clauses $D_0$, which are logically equivalent to the original input clauses $I$ up to level-0 simplifications, and the temporary (learnt) clauses $D_\infty$, each of which is implied by $D_0$ at level~0.

Table~\ref{tab:clauses-after} and Table~\ref{tab:clauses-notation} present the clause classification from the first invocation of \inprocess onward, generalizing and extending this standard partition. 

Specifically, in addition to active-pervasive and temporary clauses, we track inactive \emph{stashed} clauses and active \emph{deputy} clauses. Before restoration, each stashed clause is always represented by some active deputy clause, although this deputy may change over time.

Each active clause group $D_k \subseteq A$ and each stashed clause group
$S_k^q \subseteq S$ is indexed by a level~$k$.
Clauses in $D_k$ remain active and \emph{must not be removed} until the
solver backtracks below~$k$.
Once the solver backtracks below~$k$, this restriction is lifted, and
clauses in $D_k$ may either be deleted or treated as temporary clauses,
that is, moved to $D_\infty$.
Since the solver is always below~$\infty$, clauses in $D_\infty$ are,
by convention, always eligible for deletion.

Clauses in a stashed group $S_k^q$ remain inactive until the solver backtracks below their associated level~$k$. When this occurs, they are restored into $D_q$, where we refer to $q < k$ as their \emph{restoration level}.

\subsection{Invariants}

\begin{figure}[t]
  \centering
  \begin{tcolorbox}[colback=gray!12,
                    colframe=gray!60!black,
                    boxrule=0.4pt,
                    arc=2pt,
                    left=6pt,right=6pt,top=6pt,bottom=6pt,
                    width=1\linewidth,
                    title={BI invariants (see Table~\ref{tab:clauses} for notation)}]
  \begin{enumerate}
    \item \label{inv:in}\textbf{Input-equivalence: $I \leftrightarrow^0 P$.}      
   \textit{ In words:} the current set of pervasive clauses is logically equivalent to the original input clauses.

    \item \label{inv:co}\textbf{Correctness: $N \leftrightarrow^n P$.}
    \textit{ In words:} the active-non-temporary clauses are logically equivalent to the pervasive clauses until backtracking below $n$.

    \item \label{inv:re}\textbf{Representation.}
For all $\alpha \in S_{i \le n}$, one of the following holds:
\begin{enumerate}
  \item \textbf{Deputy-representation}: For some $r \le i$ and $r$-deputy $\beta \in D_r$, $\beta \to^i \alpha$.
  \item \textbf{Restored-representation}: For some $r \le i$, $j > i$, and $r$-restored $\beta \in S_j^r$, $\beta \to^i \alpha$.
  \item \textbf{Trivial-representation}: $\top \to^i \alpha$ (equivalently, $\tau^{\le i} \models \alpha$ by Def.~\ref{def:satisfaction-level-d}).
\end{enumerate}
\textit{   In words:} every clause stashed at some level $i \le n$ is implied at level $i$ either by a deputy whose level is at most $i$, or by a clause stashed at some level $j > i$ that will be restored as such a deputy, or else is trivially implied at level~$i$.

    \item \label{inv:cl}\textbf{Closure: $\forall \gamma \in (A \cup S).\, P \to^0 \gamma$.}
    \textit{ In words:} every clause, active or stashed, is implied by the pervasive clauses.
  \end{enumerate}
  \end{tcolorbox}
\caption{BI invariants (see Table~\ref{tab:clauses-notation} for notation). 
The level $n$ is the decision level of the latest completed BI invocation (\inprocess or \backtrack). 
The connectives $\to^i$ and $\leftrightarrow^i$ denote implication and equivalence at decision level $i$ (Def.~\ref{def:connectives-level-d}). 
Inv.~\ref{inv:co} follows from the other invariants (Lemma~\ref{lem:co-from-re-cl}).}

  \label{fig:bi-invariants}
\end{figure}

Fig.~\ref{fig:bi-invariants} summarizes the invariants jointly maintained by the main SAT search engine and our BI module. The BI module is explicitly designed to enforce these invariants, while they serve as requirements for the search engine. We will shortly argue that these requirements are natural for a standard SAT solver.

\subsubsection{High-Level Invariants and Framework Correctness}
Invs.~\ref{inv:in} and~\ref{inv:co} together guarantee the correctness of our framework: at any point during the search, the active-non-temporary clauses $N$ are logically equivalent to the original input clauses $I$ (up to level-0 simplifications implied by $I$). Strictly speaking, Inv.~\ref{inv:co} ensures this equivalence only until backtracking below level~$n$, but maintaining this invariant during backtracking is precisely the responsibility of our \backtrack\ procedure.

\subsubsection{Low-Level Invariants}

If Invs.~\ref{inv:in} and~\ref{inv:co} capture the high-level correctness of our framework, then Invs.~\ref{inv:re} and~\ref{inv:cl} provide the corresponding low-level guarantees.

Intuitively, Inv.~\ref{inv:re} ensures that every stashed clause is represented until its restoration. Unless it is trivially implied (trivial-representation), it is either represented by an active deputy (deputy-representation) or by another stashed clause that will later be restored as such a deputy (restored-representation). In the restored-representation case, at any concrete point the stashed clause is still represented by an active deputy, but \emph{transitively}: it is represented via a chain of stashed clauses, the last of which is represented by an active deputy.

Inv.~\ref{inv:cl} ensures that all clauses in the system are implied by the pervasive clauses.

Lemma~\ref{lem:co-from-re-cl} shows that Invs.~\ref{inv:re} and~\ref{inv:cl} imply Inv.~\ref{inv:co} throughout the execution of our algorithm. The proofs of this and all other lemmas appear in Appendix~\ref{app:proofs}.

\begin{restatable}[Correctness from representation and closure]{lemma}{coFromReCl}\label{lem:co-from-re-cl}
Assume Invs.~\ref{inv:re} and~\ref{inv:cl} hold. Then Inv.~\ref{inv:co} holds as well, i.e., $N \leftrightarrow^n P$.
\end{restatable}

\subsubsection{Why the Search Engine Maintains All Invariants}
We argue that a standard SAT search engine naturally maintains our invariants, both prior to the first \inprocess call and between subsequent invocations of \inprocess and \backtrack. Sect.~\ref{sec:alg} constructively shows that our BI module maintains them as well.

We assume that the search engine itself does not apply inprocessing (which is instead handled via interaction with our BI module), and that it performs only standard level-0 simplifications implied by the input clauses $I$, such as deleting globally satisfied clauses and removing globally falsified literals. Throughout this section, all logical equivalences and implications are understood modulo these level-0 simplifications, as justified by Cor.~\ref{cor:level0-implied}.


Let us verify that our invariants (skipping Inv.~\ref{inv:co} by Lemma~\ref{lem:co-from-re-cl}) hold at the time of the first \inprocess call by the search engine. Inv.~\ref{inv:in} holds since, by construction, $D_0$ comprises the input clauses $I$ up to the standard level-0 simplifications. Inv.~\ref{inv:re} is vacuous, because no stashed clauses exist before inprocessing. Inv.~\ref{inv:cl} holds since $P$ comprises $D_0$, and $A \cup S$ comprises $D_0 \cup D_\infty$, where every temporary clause in $D_\infty$ is implied by $D_0$ by construction.

We now argue that the search engine maintains all invariants in each phase between two BI invocations (\inprocess or \backtrack). Let $n$ be the decision level of the latest BI call. Crucially, during the phase, the level never drops below $n$, since \backtrack is invoked whenever backtracking below $n$ occurs.
At the beginning of the phase, Invs. ~\ref{inv:in} and~\ref{inv:co} imply that $I$, $P$, and $N$ are equivalent at level~$n$, while $D_\infty$ is implied by these groups by Inv.~\ref{inv:cl}. The search engine operates only on the active clauses $A = N \cup D_\infty$, which, by construction, remain equivalent to $I$  at level~0 until the end of the phase. Hence Invariants~\ref{inv:in} and~\ref{inv:cl} continue to hold throughout the phase. As for Inv.~\ref{inv:re}, the stashed clause sets do not change, being inactive, while deputy clauses may either remain unchanged or be simplified at level~0 (the latter only if the solver allows level-0 simplifications at arbitrary levels). Such simplifications preserve restored- and trivial-representation. Deputy-representation is preserved, except when a representing deputy becomes satisfied at level~0 and is deleted, in which case any stashed clause it represented falls under trivial-representation. Inv.~\ref{inv:co} holds by Lemma~\ref{lem:co-from-re-cl}.

\subsection{Algorithms}
\label{sec:alg}
In this section, we introduce our BI algorithms \inprocess and \backtrack, and demonstrate that they maintain the invariants in Fig.~\ref{fig:bi-invariants}.

\subsubsection{\inprocess}
\label{sec:algin}

\begin{algorithm}[H]
\caption{\inprocess at level $n$}
\label{alg:subsumption}
\begin{algorithmic}[1]
\While{there exist distinct active clauses $\alpha,\beta \in A$ and some $l$ such that $\csubs{l}{\alpha}{\beta}$}
\label{alg:subsumption:while}
  \State Pick any active $\alpha \in D_a$ and $\beta \in D_b$ with $0 \le a,b \le \infty$, $\beta \ne \alpha$, such that $\csubsmin{k}{\alpha}{\beta}$
  \label{alg:subsumption:pick}
  \If{$b < k$}
    \label{alg:subsumption:case-b-less-k}
    \State \ifb $a > k$ \thenb Promote $\alpha$ to $D_k$
    \label{alg:subsumption:promote-k}
    \State Deactivate and stash $\beta$ by moving it to $S_k^b$
    \label{alg:subsumption:stash-beta}
  \Else \Comment{$b \ge k$}
    \label{alg:subsumption:case-b-ge-k}
    \State \ifb $a > b$ \thenb Promote $\alpha$ to $D_b$
    \label{alg:subsumption:promote-b}
    \State Delete $\beta$ \Comment{Alternatively, demote $\beta$ to a temporary clause in $D_\infty$}
    \label{alg:subsumption:delete-beta}
  \EndIf
\EndWhile
\end{algorithmic}
\end{algorithm}

Our \inprocess algorithm is shown in Alg.~\ref{alg:subsumption}. It operates in a loop. In each iteration, it selects a pair of distinct active clauses $\alpha \in D_a$ and $\beta \in D_b$ with $0 \le a,b \le \infty$ such that $\alpha$ subsumes $\beta$ at some decision level, and this subsumption is \emph{minimal} at level~$k$. (If no such pair exists, the algorithm terminates.) The goal is either to delete $\beta$ entirely whenever this is sound, or otherwise to stash $\beta$ for as long as possible—that is, with the lowest possible restoration level—while ensuring that $\beta$ remains represented until it is restored.

Before presenting the algorithm in detail, we need a notion of rank; see Def.~\ref{def:rank}.

\begin{definition}[Rank of an active clause]\label{def:rank}
The \emph{rank} of an active clause $\gamma \in D_i$ is $i$.
\end{definition}

Promoting $\gamma$ to a lower group (i.e., moving $\gamma$ from $D_i$ to some $D_j$ with $j < i$) does not break the invariants in Fig.~\ref{fig:bi-invariants}, and in particular preserves Inv.~\ref{inv:re}, as formalized in Lemma~\ref{lem:promotion-invariants}.

\begin{restatable}[Promotion preserves invariants]{lemma}{propreinv}\label{lem:promotion-invariants}
Let $\gamma \in D_i$ be an active clause, and suppose all invariants in
Fig.~\ref{fig:bi-invariants} hold. If we \emph{promote} $\gamma$ to some group
$D_j$ with $j < i$ (by moving it), then all invariants continue to hold.
\end{restatable}

We now analyze one iteration of \inprocess{} by case-splitting on the relationship between $\beta$’s rank $b$ and the minimal subsumption level $k$.

\textbf{If $\beta$'s rank $b$ is lower than $k$}
(line~\ref{alg:subsumption:case-b-less-k}), we deactivate and stash
$\beta$ at level $k$ with $b$ as its restoration level
(line~\ref{alg:subsumption:stash-beta}). Before that, if $\alpha$'s rank
is higher than $k$, we promote $\alpha$ to $D_k$
(line~\ref{alg:subsumption:promote-k}) so that it can represent
$\beta$ while the solver remains at levels $\ge k$. Furthermore, not only
$\beta$ itself then becomes represented by $\alpha$, but so does any clause
stashed at level $k$ or higher that was previously represented by $\beta$,
whereas any clause stashed at level $< k$ and previously represented by
$\beta$ is still represented by $\beta$, but now via restored-representation.
Lemma~\ref{lem:b-less-k-case} formalizes that these updates preserve all BI
invariants.

\begin{restatable}[Soundness of the $b<k$ case]{lemma}{blesskcase}\label{lem:b-less-k-case}
In the case $b < k$, the updates in
lines~\ref{alg:subsumption:promote-k}--\ref{alg:subsumption:stash-beta}
preserve all BI invariants in Fig.~\ref{fig:bi-invariants}.
\end{restatable}

Otherwise, \textbf{if $\beta$'s rank $b$ is not lower than $k$}, we promote
$\alpha$ to $D_b$ if its rank $a$ is not already $\le b$ and delete $\beta$
(or demote $\beta$ to a temporary clause in $D_
\infty$).
This is safe because $\alpha$ then represents both $\beta$ and any clause
previously represented by $\beta$ until backtracking below $b$. Once the solver
backtracks below $b$, Cor.~\ref{cor:no-repr-below-rank} implies that $\beta$
no longer represents any clause. Hence, $\beta$ can be safely deleted. Cor.~\ref{cor:no-repr-below-rank}
follows from Def.~\ref{def:rank} and Inv.~\ref{inv:re}.

\begin{corollary}[No representation below rank]\label{cor:no-repr-below-rank}
An active clause of rank $k$ can only represent stashed clauses at levels
$\ge k$.
\end{corollary}

Lemma~\ref{lem:b-ge-k-case} shows that these updates preserve all BI invariants.

\begin{restatable}[Soundness of the $b \ge k$ case]{lemma}{bgekcase}\label{lem:b-ge-k-case}
In the case $b \ge k$, the updates in
lines~\ref{alg:subsumption:promote-b}--\ref{alg:subsumption:delete-beta}
preserve all BI invariants in Fig.~\ref{fig:bi-invariants}.
\end{restatable}

\subsubsection{\backtrack}
\label{sec:algba}

\begin{algorithm}[H]
\caption{\backtrack to level $n$}
\label{alg:backtrack}
\begin{algorithmic}[1]
\State Let $n$ be the current decision level (after backtracking).
\State Let $m \gets \max\{\, i \mid D_i \ne \emptyset \ \text{or}\ S_i \ne \emptyset \,\}$ \Comment{Highest level with active or stashed clauses}
\For{$i \gets m$ \textbf{downto} $n+1$}
  \State \textbf{for all} $\alpha \in D_i$ \textbf{do} Delete $\alpha$ \Comment{Alternatively, demote $\alpha$ to a temporary clause in $D_\infty$}
  \State \textbf{for all} $q \in \{0,\ldots,i-1\}$ \textbf{do}
    \State \quad \textbf{for all} $\beta \in S_i^q$ \textbf{do} Activate $\beta$ and move it from $S_i^q$ to $D_q$
\EndFor
\end{algorithmic}
\end{algorithm}

Our backtracking algorithm \backtrack is shown in Alg.~\ref{alg:backtrack}.
Given the current decision level $n$ after the SAT solver has backtracked,
the procedure identifies the highest level $m$ at which there are still
active or stashed clauses, and then iterates from $m$ down to $n+1$. For each
visited level $i$, it first deletes all deputy clauses in $D_i$, which is
sound (w.r.t.\ Inv.~\ref{inv:re}) because no clause of rank $i$ can
represent any stashed clause below its rank (Cor.~\ref{cor:no-repr-below-rank}). Alternatively, such deputies may be
demoted to temporary clauses, without affecting correctness (Inv.~\ref{inv:cl}).
It then restores every clause stashed at level $i$ with restoration
level $q$ by moving it into $D_q$, which is safe w.r.t.\ Inv.~\ref{inv:re}
because any stashed clause previously restored-represented via a freshly restored
clause continues to be represented by it, now via deputy-representation.

Lemma~\ref{lem:one-step-backtrack} formally argues that a single iteration of the loop in
Alg.~\ref{alg:backtrack} (i.e., processing one level $i$ by deleting $D_i$
and restoring all $S_i^q$) preserves all BI invariants in
Fig.~\ref{fig:bi-invariants}. By induction over the iterations
$m,m-1,\ldots,n+1$, this yields the correctness of \backtrack.

\begin{restatable}[One-step backtrack preserves invariants]{lemma}{onestepback}\label{lem:one-step-backtrack}
Processing a single level $i$ in \backtrack{} (i.e., deleting $D_i$ and restoring
all $S_i^q$ into $D_q$) preserves all BI invariants in
Fig.~\ref{fig:bi-invariants}.
\end{restatable}

\section{Backtrackable Inprocessing}
\label{sec:bi}

This section extends the subsumption-only non-incremental BI algorithm to support selfsumption, BVE, and incremental SAT solving. With the theoretical foundations (clause classification, invariants, correctness) already established in Sect.~\ref{sec:bisub}, we focus primarily on the practical aspects of the algorithms.

{\sloppy Our high-level inprocessing algorithm largely reuses \emph{incremental SatELite}—the original global level-0 non-incremental SatELite preprocessing algorithm~\cite{DBLP:conf/sat/EenB05}, extended to global level-0 \emph{incremental} preprocessing in our earlier work~\cite{DBLP:conf/sat/NadelRS12}. Sect.~\ref{sec:related-incremental} reviews incremental SatELite and related work, and Sect.~\ref{sec:bicomplete} presents our complete BI algorithm.\par}

\newcommand{\occ}{\mathrm{occ}}     
\newcommand{\CH}{\mathcal{H}}       
\newcommand{\EO}{\mathcal{O}}       
\newcommand{\ES}{\mathcal{E}}       

\subsection{Incremental Preprocessing with SatELite and Related Work}
\label{sec:related-incremental}

\newcommand{\resolvent}{\otimes}

We start with standard definitions for resolution and selfsumption (self-subsuming resolution).

\begin{definition}[Resolution]\label{def:resolution}
Let $\alpha,\beta$ be clauses and $x$ a variable. They are \emph{resolvable on $x$}
if $x\in\alpha$ and $\lnot x\in\beta$. In this case the \emph{resolvent} is $\alpha \resolvent_x \beta := (\alpha\setminus\{x\}) \cup (\beta\setminus\{\lnot x\})$, and the \emph{resolution rule} derives $\alpha \resolvent_x \beta$ from
$\alpha$ and $\beta$.
\end{definition}

\begin{definition}[Selfsumption, Strengthening]\label{def:selfsumption}
Let $\alpha,\beta$ be clauses resolvable on $x$. We say that
\emph{$\alpha$ selfsumes $\beta$ on $x$} if
$\alpha \resolvent_x \beta \subseteq \beta$.
Equivalently, writing $\alpha = C \lor x$ and $\beta = D \lor \lnot x$, this holds
iff $C \subseteq D$. In this case, $\beta$ may be \emph{strengthened} to
$D = \alpha \resolvent_x \beta$.
\end{definition}

For example, if $\alpha = (a \lor x)$ and $\beta = (a \lor b \lor \lnot x)$,
then the resolvent on $x$ is $\alpha \resolvent_x \beta = (a \lor b)$, which is
a subset of $\beta$. Hence $\alpha$ selfsumes $\beta$ on $x$, and $\beta$ can be
strengthened to $(a \lor b)$.

Finally, we use the term \emph{sumption} to refer to subsumption and
selfsumption together, and the verb \emph{to sume} to indicate that a clause
either subsumes or selfsumes another clause.

Alg.~\ref{alg:satelite} presents incremental preprocessing with
SatELite~\cite{DBLP:conf/sat/EenB05,DBLP:conf/sat/NadelRS12} in high-level, using the core
data structures listed in Fig.~\ref{fig:incsatelite}. 

Consider the data structures in Fig.~\ref{fig:incsatelite}.
The occurrence lists $\occ[\ell]$ store all clauses containing literal~$\ell$.
The clause heap $\CH$ is a size-based heap used for efficient backward
sumption to scan candidate clauses from smallest to largest to see whether
they sume others.
Finally, all eliminated variables are maintained in their elimination order in
$\EO$ (preserving this order is crucial for incrementality~\cite{DBLP:conf/sat/NadelRS12}), and for each eliminated variable $v$ the original clauses are stored in $F^+(v)$ and $F^-(v)$.

Alg.~\ref{alg:satelite} is invoked at the beginning of every incremental SAT
query. It first scans the clauses to initialize $\occ$ and $\CH$
(line~\ref{line:init}). Then, the loop in line~\ref{line:eo-loop} iterates
over all eliminated variables in $\EO$. If new clauses containing a variable $v$
have appeared since the previous invocation, the algorithm chooses to either reintroduce
$v$ by restoring $F^+(v)$ and $F^-(v)$ (optionally removing their resolvents) and removing $v$ from $\EO$,
or reeliminate $v$ by adding the new clauses to $F^+(v)$ and $F^-(v)$ and
creating new non-tautological resolvents, using the standard bounded-clause elimination heuristic
to make this choice. Both reintroduction and reelimination may generate new clauses with eliminated
variables, but, crucially for correctness, only with variables that will still
be visited later in the current pass over $\EO$~\cite{DBLP:conf/sat/NadelRS12,DBLP:conf/sat/FazekasBS19}.
The algorithm then enters the main preprocessing loop (line~\ref{line:outer-repeat}).
It begins by applying backward sumption to a fixpoint (line~\ref{line:sumption}).
Next, every variable that occurs in clauses removed or updated during the latest
sumption step is considered for elimination (line~\ref{line:elim-scan}).
A variable $v$ is eliminated as long as doing so does not increase the clause
count (line~\ref{line:budget}). Finally, elimination is applied only to
pervasive clauses: temporary clauses are ignored, and, if elimination succeeds,
any temporary clauses containing $v$ are deleted.

\subsubsection{Other Related Work}

We reuse the high-level structure of incremental SatELite in our BI framework. 
Supporting gate extraction~\cite{DBLP:conf/sat/EenB05}, as well as extending 
BI beyond sumption and variable elimination (in the spirit of 
\cite{DBLP:conf/sat/FazekasBS19}, which achieves this for global level-0 
inprocessing), is left for future work. We previously proposed incremental 
preprocessing after propagating assumptions in 
\cite{DBLP:conf/sat/NadelRS14}. However, the present work is both 
significantly more general—our framework supports \emph{in}processing at any 
decision level and at any point during solving, rather than only 
\emph{pre}processing at the beginning of each incremental query—and 
practically more useful, since~\cite{DBLP:conf/sat/NadelRS14} is 
closed-source (whereas our implementation is open-source) and relies on 
maintaining resolution proofs in memory, an infrastructure not supported by 
modern SAT solvers.


\begin{figure}[t]
\centering
\begin{tcolorbox}[colback=gray!6,colframe=gray!50!black,boxrule=0.4pt,arc=2pt,
  left=6pt,right=6pt,top=6pt,bottom=6pt,
  title={Core data structures for incremental preprocessing with SatELite~\cite{DBLP:conf/sat/EenB05,DBLP:conf/sat/NadelRS12}}]

\begin{itemize}
  \item \emph{Per-literal occurrence lists} $\occ[\ell]$: vector of clauses containing literal~$\ell$.
  \item \emph{Clause heap} $\CH$: size-based clause heap for backward sumption.
  \item \emph{Eliminated variables:}
  \begin{enumerate}[label=(\roman*), leftmargin=*]
    \item eliminated variables in their elimination order $\EO=[v_1,\ldots,v_m]$;
    \item for each $v\in\EO$, original clauses 
      $F^+(v)=\{C\mid v\in C\}$ and $F^-(v)=\{C\mid \lnot v\in C\}$.
  \end{enumerate}
\end{itemize}

\end{tcolorbox}
\caption{Core data structures used by incremental SatELite.}
\label{fig:incsatelite}
\end{figure}


\begin{algorithm}[t]
\caption{Incremental Preprocessing with SatELite}
\label{alg:satelite}
\begin{algorithmic}[1]

\State Initialization: scan clauses to initialize $\occ$ and $\CH$. \label{line:init}

\For{\textbf{each} $v \in \EO$} \label{line:eo-loop}\Comment{Reintroduction/reelimination loop}
  \If{new clauses with $v$ exist} \label{line:eo-exists}
     \State Reintroduce or reeliminate $v$ \label{line:eo-repair}\Comment{Reintroduction removes $v$ from $\EO$}
  \EndIf
\EndFor

\Repeat \label{line:outer-repeat} \Comment{Main preprocessing loop}
  \State Apply backward \emph{sumption} (subsumption + selfsumption) to a fixpoint. \label{line:sumption}

  \For{\textbf{each} variable $v$ in clauses removed or updated during latest sumption} \label{line:elim-scan}
    \If{eliminating $v$ does not increase the clause count} \label{line:budget}\Comment{ignoring temporaries}
      \State Eliminate $v$: save $F^+(v)\cup F^-(v)$ and add non-tautological resolvents. \label{line:eliminate}
      \State Append $v$ to $\EO$. \label{line:update-stores}
    \EndIf
  \EndFor

\Until{no elimination in the last iteration} \label{line:repeat-end}

\end{algorithmic}
\end{algorithm}

\subsection{Incremental Backtrackable Inprocessing with Sumption and BVE}
\label{sec:bicomplete}

Throughout the rest of the paper, whenever we refer to level-0 notions
(satisfaction, implication, equivalence, or subsumption with respect to $\tau^{\le 0}$),
we call them \emph{global}. In contrast, notions at levels $d>0$ (i.e., with respect
to $\tau^{\le d}$ for some $d>0$) are called \emph{conditional}, since they depend on
the current trail.

To integrate our subsumption-only BI framework from Section~\ref{sec:bisub} into the SAT solver, it is sufficient to:
\begin{enumerate}
  \item Apply \inprocess (Alg.~\ref{alg:subsumption}) for handling conditional subsumption directly as the implementation of the sumption step in incremental SatELite (line~\ref{line:sumption} in Alg.~\ref{alg:satelite}), retaining the original processing order from smallest to largest clauses via the size-based heap, but ignoring conditionally falsified literals when computing clause sizes.
  \item Stay synchronized with respect to backtracking: let $m$ be the decision level at the end of the latest \inprocess or \backtrack (Alg.~\ref{alg:backtrack}) invocation. Then \backtrack \emph{must} be invoked whenever the solver backtracks below level~$m$.
\end{enumerate}

Below, Sect.~\ref{sec:sumption} defines conditional selfsumption and explains how to extend the above BI integration in incremental SatELite to support it, while Sects.~\ref{sec:sat} and~\ref{sec:bve} show how to handle conditionally satisfied clauses and BVE, respectively, within our framework. Finally, Sect.~\ref{sec:bistrategies} discusses the integration of incremental SAT solving (which comes for free in our framework) and introduces several strategies for applying BI for incremental preprocessing.

\subsubsection{Handling Selfsumption}
\label{sec:sumption}

We begin by defining conditional $d$-selfsumption and min-$d$-selfsumption,
in analogy to $d$-subsumption (Def.~\ref{def:d-subsumption})
and min-$d$-subsumption (Def.~\ref{def:minimal-d-subsumption}).

\begin{definition}[$d$-selfsumption]\label{def:d-selfsumption}
Let $\alpha,\beta$ be clauses resolvable on $x$ and
$d \in \mathbb{N}$. We say that \emph{$\alpha$ $d$-selfsumes $\beta$ on $x$},
written $\mathrm{selfsume}^d_x(\alpha,\beta)$, if $(\alpha \resolvent_x \beta)^{\le d} \subseteq \beta^{\le d}$.
\end{definition}

\begin{definition}[min-$d$-selfsumption]\label{def:minimal-d-selfsumption}
Let $\alpha,\beta$ be clauses resolvable on $x$ and
$d \in \mathbb{N}$. We say that \emph{$\alpha$ minimally $d$-selfsumes $\beta$
on $x$}, written $\mathrm{selfsume}^{\min(d)}_x(\alpha,\beta)$, if:
\begin{enumerate}
  \item $\mathrm{selfsume}^d_x(\alpha,\beta)$ holds, and
  \item for every decision level $d'$ strictly smaller than $d$,
        $\mathrm{selfsume}^{d'}_x(\alpha,\beta)$ does not hold.
\end{enumerate}
Thus $d$ is the minimal decision level at which $\alpha$ $d$-selfsumes $\beta$ on $x$.
\end{definition}

We have augmented Alg.~\ref{alg:subsumption} to support $d$-selfsumption 
as follows:
\begin{enumerate}
\item We identify min-$d$-selfsumption alongside min-$d$-subsumption.
\item If $d=0$ (selfsumption is global), it is handled by strengthening as in global inprocessing.
\item Whenever conditional $\mathrm{selfsume}^{\min(d)}_x(\alpha,\beta)$ for $d>0$ is
      discovered, we reduce its handling to $d$-subsumption as follows. We create
      a new clause marked as \emph{temporary}, $\gamma \in D_\infty$, with
      $\gamma := \alpha \resolvent_x \beta$. By construction,
      $\csubsmin{d}{\gamma}{\beta}$ holds, hence we let our $d$-subsumption algorithm
      handle this case.
\end{enumerate}

We emphasize that, unlike in global selfsumption, in conditional selfsumption we cannot soundly strengthen the selfsumed clause by removing a literal; instead, we must create a new clause. This is due to the presence
of conditionally falsified literals in the selfsuming $\alpha$ but not in the selfsumed $\beta$, rendering the
strengthening justified only conditionally and not globally. For example, let
$\alpha \in A$, $\alpha = (a \lor b \lor x)$, and $\beta \in N$, $\beta = (b \lor c \lor \lnot x)$ under the trail
$\tau = \langle \lnot a @ 1 \rangle$. We have
$\alpha \resolvent_x \beta = (a \lor b \lor c)$, so
$(\alpha \resolvent_x \beta)^{\le 1} = (b \lor c)$ and
$\beta^{\le 1} = (b \lor c \lor \lnot x)$, hence
$(\alpha \resolvent_x \beta)^{\le 1} \subseteq \beta^{\le 1}$ and thus
$\alpha$ min-1-selfsumes $\beta$ on $x$. Our $d$-selfsumption mechanism therefore creates a new temporary clause
$\gamma \in D_\infty$ with $\gamma := \alpha \resolvent_x \beta = (a \lor b \lor c)$
and, given that $\csubsmin{1}{\gamma}{\beta}$, delegates the rest to
$d$-subsumption, which stashes $\beta$ at level~1 and promotes $\gamma$ to $D_1$
to represent $\beta$. Clearly, removing $\lnot x$ from $\beta$ would be unsound.

As a practical note, since creating new clauses is significantly more expensive than removing literals from existing ones, our implementation applies conditional selfsumption only when the number of unassigned literals in the selfsumed clause is at most~3.

\subsubsection{Handling Conditionally Satisfied Clauses}
\label{sec:sat}
Unlike in global inprocessing, backtrackable inprocessing may encounter clauses that are \emph{conditionally} satisfied under the current trail and thus irrelevant at the current decision level. To handle such clauses, we first define minimal satisfaction (w.r.t.\ the implicit trail~$\tau$).

\begin{definition}[Minimal satisfaction]\label{def:minimal-satisfaction}
Let $\varphi$ be a Boolean formula.
We say that $\varphi$ is \emph{minimally satisfied at level $d$} under $\tau$
if:
\begin{enumerate}
  \item $\varphi$ is satisfied at level $d$ under $\tau$, i.e., $\tau^{\le d} \models \varphi$, and
  \item for every decision level $d'$ with $0 \le d' < d$, we have
        $\tau^{\le d'} \not\models \varphi$.
\end{enumerate}
In this case, $d$ is the minimal decision level at which $\varphi$ is
satisfied under $\tau$.
\end{definition}

We handle conditionally satisfied clauses as follows.
During the inprocessor’s initialization (line~\ref{line:init} of
Alg.~\ref{alg:satelite}), we do \emph{not} insert them into the main data
structures (the occurrence lists $\occ$ and the clause heap $\CH$).
Instead, taking advantage of our clause classification, whenever a clause
$\beta \in D_i$ is conditionally satisfied under the current trail, we
compute its minimal satisfaction level $d$ (in the sense of
Def.~\ref{def:minimal-satisfaction}), remove $\beta$ from the active
set $D_i$, and stash it in $S_d^i$ to be restored once the solver has backtracked below $d$, by which time $\beta$ will no longer be satisfied.
Note that Inv.~\ref{inv:re} is maintained for any such clause due to
trivial-representation.

\subsubsection{Handling Bounded Variable Elimination (BVE)}
\label{sec:bve}

We now turn to Bounded Variable Elimination (BVE). Unlike for sumption, there is no ``conditional'' variable elimination, since the variable elimination algorithm itself does not depend on the trail. However, the presence of conditionally satisfied clauses is promising
for \emph{bounded} variable elimination \emph{efficiency-wise}: skipping such clauses as described in  Sect.~\ref{sec:sat} increases the chances of eliminating more variables while adhering to the clause-count bound.
For correctness, when those clauses are later restored, the relevant eliminated variables must be
reintroduced or reeliminated, but this is handled automatically as laid out
below.

Our BI integration preserves the standard BVE scheme of incremental SatELite, but, as explained below, it introduces six easy-to-implement, BI-specific changes, mostly related to assigning and managing ranks of original clauses and resolvents. Below, these changes are numbered as \textbf{(Ci)} at the end of the sentences introducing them.

In BI, new clauses containing eliminated variables may arise not
only from the user (as in standard incremental SAT), but also from restoring
stashed clauses. Fortunately, incremental SatELite (Alg.~\ref{alg:satelite}) already handles clauses with eliminated variables regardless of their origin. We just need to ensure that the search algorithm
invokes inprocessing whenever even one clause with an eliminated variable
is present in the solver, whether it is a user clause added during a new incremental query or a restored stashed clause \bvechange.

The remaining main question for handling BVE in BI is how to treat deputy clauses
during elimination and backtracking. Recall that a deputy of rank $i$ (that is,
a deputy in $D_i$) is required only as long as the solver does not backtrack
below $i$, at which point all clauses in $D_i$ are deleted by \backtrack
(Alg.~\ref{alg:backtrack}). Therefore, conceptually, any resolvents produced
using such deputies become obsolete as soon as the solver backtracks below the
highest rank of their parents. To reflect this, we assign each resolvent the
\emph{maximum} rank of its parent clauses as follows: when, while eliminating
a variable $v$, we resolve $\alpha \in D_i$ with
$\beta \in D_j$ on $v$, producing the resolvent $\gamma = \alpha \resolvent_v \beta$, we assign $\gamma$ to $D_k$, where $k := \max(i,j)$ \bvechange.
 Furthermore, we keep the original rank of every clause in $F^+(v)$ and $F^-(v)$ for any
eliminated variable $v$ \bvechange. When the solver backtracks below level $i$, all clauses
of rank $i$ are deleted from $F^+(v)$ and $F^-(v)$ for every eliminated variable
$v$ \bvechange. When a variable $v$ is reintroduced, the remaining clauses in $F^+(v)$ and
$F^-(v)$ keep their rank when activated \bvechange. Finally, as BI is designed to solve the problem under the trail, BVE attempts to eliminate only unassigned variables, ignoring those already assigned on the trail \bvechange.

Furthermore, following prior work on incremental inprocessing~\cite{DBLP:conf/sat/NadelRS12,DBLP:conf/sat/FazekasBS19},
we restrict elimination to variables occurring in newly added clauses; in BI, we
further consider variables appearing in clauses that contain literals assigned on
the current trail, thereby avoiding repeated elimination attempts for variables
whose relevant context has not changed.

\subsubsection{Incrementality and BI Strategies for Incremental Preprocessing}
\label{sec:bistrategies}

BI inherently accommodates clause additions and changing assumptions, so our framework requires no modifications to support incremental SAT solving and can be applied as is.

Although BI supports any number of inprocessing steps at arbitrary levels, we focus on studying the effects of \emph{incremental preprocessing} after the assumptions are propagated—a minimal step that already distinguishes BI from standard global incremental inprocessing, both in terms of algorithmic capabilities and in its empirical impact on Bounded Model Checking (BMC). Below we present several strategies for incremental preprocessing.

{\sloppy We call the first incremental preprocessing strategy enabled by BI
the \textbf{Assumption-Level (AL) strategy}.
It applies inprocessing immediately after propagating the assumptions for every incremental SAT query.
As a baseline for comparison, we also implement a global-level \textbf{\gl strategy} that backtracks to 0 and preprocesses there with the assumptions frozen (i.e., not allowed to be eliminated) at the start of every incremental query (as in CaDiCaL~\cite{DBLP:conf/cav/BiereFFFFP24}).\par}

Initial experiments on BMC instances showed that,
although the BI-enabled \al is often more efficient, there are benchmarks
where the standard global-level \gl wins. Further analysis revealed that the
reason lies in how variables of literals implied by assumptions, and clauses that become
conditionally satisfied under assumptions, are treated.

\al ignores assigned variables and stashes conditionally satisfied clauses before inprocessing, which can be beneficial because less data is processed and BVE can eliminate more variables. However, there is a trade-off. First, eliminating variables that are about to be implied by assumptions can be helpful, and this is carried out by \gl but not \al. Second, \al misses selfsumptions by conditionally satisfied clauses that could remove falsified literals from other clauses, as illustrated below. Let $\alpha = (a \lor b \lor c)$ and $\beta = (\lnot a \lor b \lor c)$ be clauses handled by \al under the trail $\tau = \langle a @ 1 \rangle$. Since $\alpha$ is stashed by BI as satisfied, \al misses the opportunity to strengthen $\beta$ to $(b \lor c)$. While this does not affect reasoning under the current trail, any new temporary clauses that depend on $\beta$ (including clauses selfsumed by $\beta$ during inprocessing and learnt clauses during search) must carry $\lnot a$, inflating the clause database and causing cache misses. In contrast, \gl would have strengthened $\beta$.

Hence, we introduce a combined \textbf{\combined strategy}: immediately after a new SAT query, we first run inprocessing at level 0 (\gl) and then again after assigning and propagating assumptions (\al). This layered approach, despite its own overhead, can sometimes simplify the formula more effectively by first applying global and then conditional simplifications.

Finally, we observed a clear pattern in when \al, \gl, and \combined work best. Let \(M\) be a measure of how intensively BI is used while \al is running.
 By default, we define $M = \frac{\text{number of conditional selfsumptions}}{\text{number of inprocessing invocations}}$ (a slightly less effective variant is $\frac{\text{stashed clause restorations}}{\text{number of conflicts}}$). Empirically, on instances with a small $M$, \al tends to be most effective (it is cheap while still having noticeable impact); for medium values of $M$, \gl is preferable (the additional overhead of \al is not justified); and when $M$ is large, \combined works best (it is then worthwhile to invest more in inprocessing, given its strong overall effect).

Hence, we designed the following dynamic \textbf{\dynamic strategy}. We start with \al and, after \num{10000} conflicts, permanently select a strategy based on the value of $M$: if $M < \num{200}$, we stay with \al; if $\num{200} \leq M < \num{3000}$, we switch to \gl; and if $M \geq \num{3000}$, we switch to \combined.

\section{Experimental Results}
\label{sec:exp}

We implemented BI in the Island SAT solver, a fork of IntelSAT~\cite{NadelIntelSAT2022}.

We started from the 300 bit-level benchmarks of the Hardware Model Checking
Competition 2017~\cite{BiereVanDijkHeljanko-FMCAD17}. We unrolled each benchmark to bound~100 using the \texttt{aigbmc} tool~\cite{Biere-FMV-TR-11-2},  to obtain incremental SAT instances, and dumped them in IntelSAT's incremental SAT format (an extension of DIMACS)~\cite{NadelIntelSAT2022}. For two of the
benchmarks, unrolling ran out of memory beyond mid-50 bounds, so we excluded
these benchmarks and used the remaining 298 in our experiments. To simulate BMC, we instructed the SAT solver to stop after the first satisfiable query, if any. 

For all experiments, we used machines with 256\,GB of memory and Intel\textsuperscript{\textregistered} Xeon\textsuperscript{\textregistered} E5-2690~v4 processors running at up to 3.5\,GHz, and we set the timeout to one hour. The solver and the benchmarks are publicly available~\cite{BI-artifacts}.
 
We compare the three BI-enhanced incremental preprocessing strategies \al, \combined, and \dynamic against the baseline \gl, which implements standard global incremental preprocessing, in terms of the number of solved bounds. To filter out trivially solved, uninteresting bounds, we report the delta relative to the vanilla Island without inprocessing.

\begin{figure}[htbp]
  \centering
  \begin{minipage}[t]{0.48\linewidth}
    \centering
    \includegraphics[width=\linewidth]{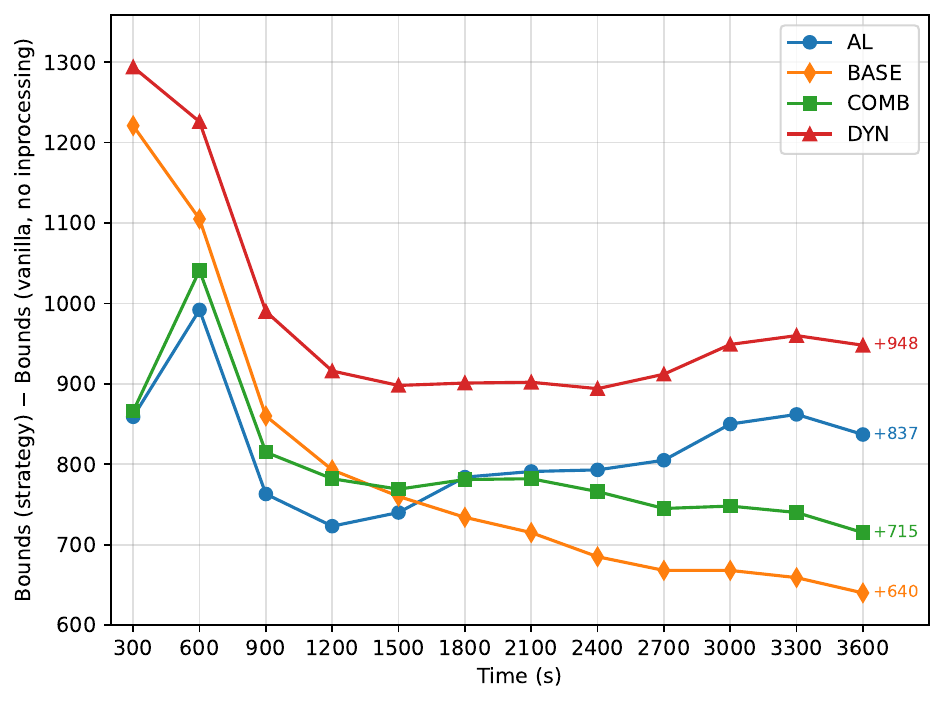}
    \captionof{figure}{Bounds solved over time per strategy: delta vs.\ vanilla (no inprocessing).}
    \label{fig:queries-delta}
  \end{minipage}\hfill
  \begin{minipage}[t]{0.48\linewidth}
    \centering
    \includegraphics[width=\linewidth]{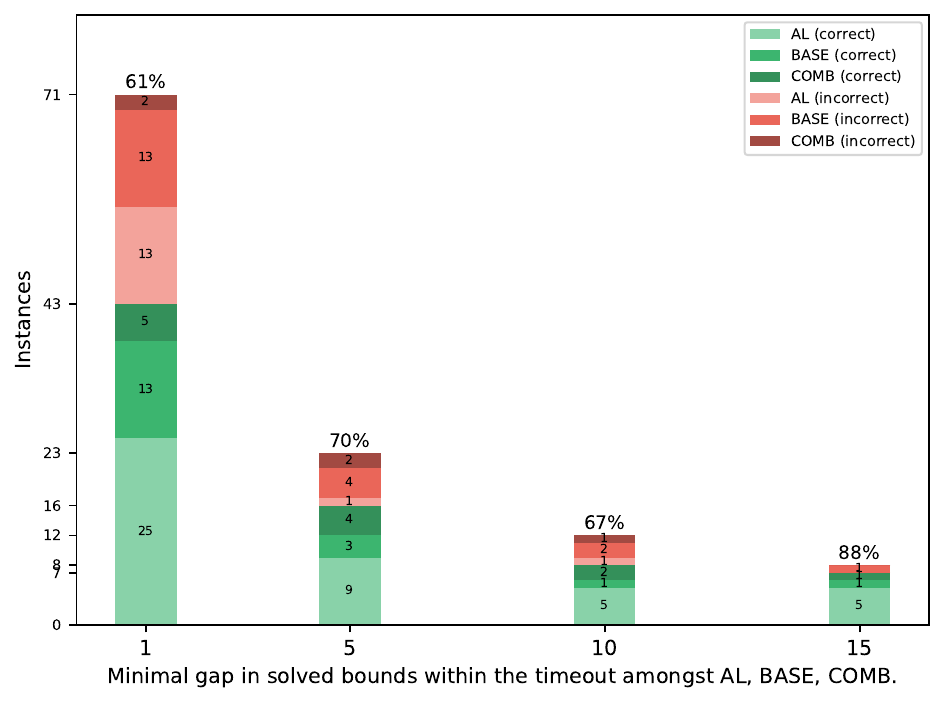}
    \captionof{figure}{\dynamic quality per minimal gap in solved bounds amongst \al, \gl, \combined.}
    \label{fig:dyn-correctness-by-gap}
  \end{minipage}
\end{figure}

The main result is shown in Fig.~\ref{fig:queries-delta}, which reports, over time, the total number of bounds solved per strategy (delta vs.\ vanilla). Although the simplest BI strategy \al (that performs inprocessing after assumptions are propagated) is initially outperformed by the baseline \gl, \al eventually overtakes \gl and, at the timeout, solves 197 more bounds. \dynamic is the superior strategy: it is more efficient than \gl from the very beginning with the gap widening to 308 more bounds than \gl and 111 more bounds than \al at the timeout. \combined is not as effective as the two other BI-enabled strategies overall.

Consider \cref{fig:dyn-correctness-by-gap}, which analyzes the quality of \dynamic's strategy choices. Each vertical bar at $x = n$ shows the decision quality of \dynamic on contested instances where \al, \gl, and \combined achieve different numbers of solved bounds; here, $n$ is the gap between the highest and lowest of those counts, so larger $n$ means higher stakes for a correct decision by \dynamic. The percentage above each bar is the fraction of those instances on which \dynamic chooses correctly. Within each bar, green segments correspond to a correct choice (a strategy that achieves the maximum number of solved bounds), and red segments to an incorrect choice; sub-segments indicate which strategy was chosen (\al, \gl, or \combined). As the gap increases, \dynamic's correctness rate rises and reaches 88\% for the largest gap. Thus, \dynamic tends to choose correctly precisely when the payoff—the difference between the best and worst static strategy—is highest. Although \combined is outperformed by the other static strategies overall, it is still sometimes chosen by \dynamic even for the highest $n$.

\Cref{fig:elims-sumptions} compares \al with \gl in terms of eliminations per variable and sumptions per clause; each point is an instance, and its color indicates which strategy solved more bounds on that instance (green: \al; red: \gl; purple: tie). Both plots use a logarithmic scale. As expected, \al's operation on a smaller formula (under the current trail) yields more eliminations and, even more so, more sumptions than \gl, so that most points lie above the diagonal. However, this is not always the case, which illustrates the trade-off described in \cref{sec:bistrategies}: \gl can eliminate variables  about to be assigned by assumption propagation and also use them for selfsumption. Color analysis shows a weak correlation between the metric and the winner  in the eliminations plot: above the diagonal, about \(21\%\) of instances are won by \al versus \(15\%\) by \gl, whereas below it only about \(4\%\) are won by \al compared to \(10\%\) by \gl (the rest are ties). For sumptions, the correlation is weaker.

\begin{figure}[htbp]
  \centering
  \begin{subfigure}[t]{0.48\linewidth}
    \centering
    \includegraphics[width=\linewidth]{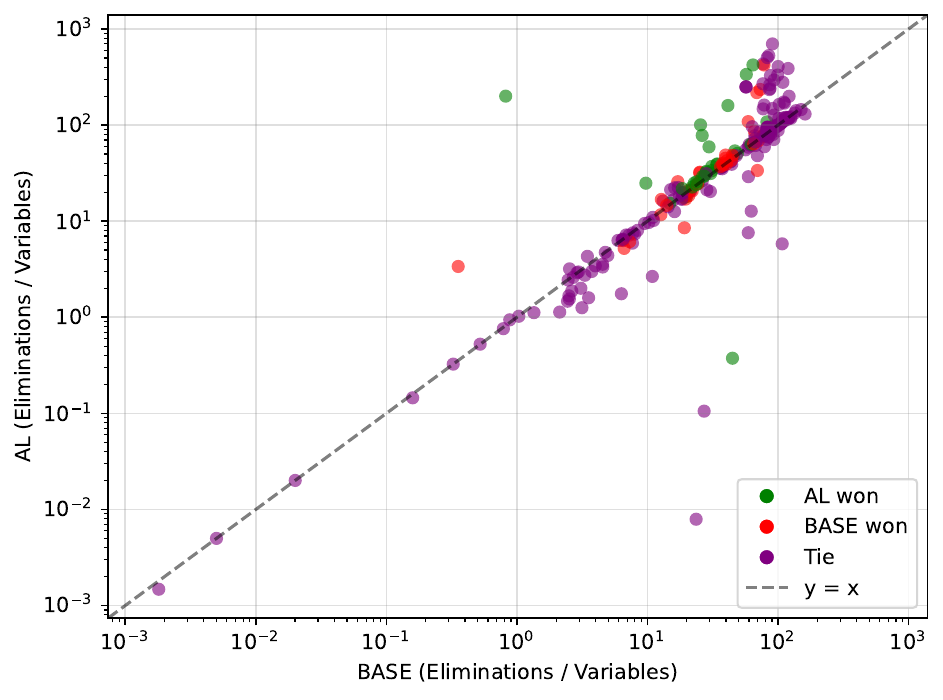}
    \caption{Eliminations/Vars: \al vs \gl}
    \label{fig:elims-vars}
  \end{subfigure}\hfill
  \begin{subfigure}[t]{0.48\linewidth}
    \centering
    \includegraphics[width=\linewidth]{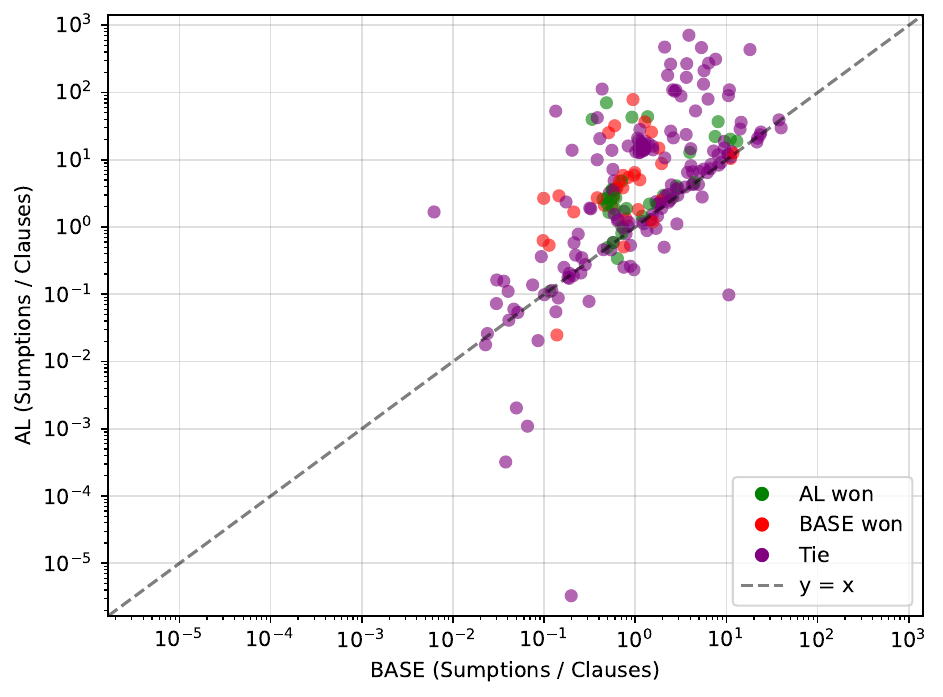}
    \caption{Sumptions/Clss: \al vs \gl.}
    \label{fig:sumptions-clauses}
  \end{subfigure}
  \caption{\al vs \gl: (a) eliminations per variable; (b) sumptions per clause, per benchmark (log scale). Point color: green = \al solved more bounds, red = \gl, purple = tie.}
  \label{fig:elims-sumptions}
\end{figure}

\section{Conclusion}
\label{sec:conclusion}

We introduced \emph{Backtrackable Inprocessing} (BI), a framework that allows inprocessing under the current trail at any level, at any point in the search, while preserving correctness under backtracking (focusing on subsumption, self-subsuming resolution, and BVE). We integrated BI into the Island solver and evaluated several strategies for incremental preprocessing  on BMC benchmarks from HWMCC'17, obtained by unrolling each benchmark up to bound~100 and using a one-hour timeout. Our best BI-enhanced strategy outperforms the baseline strategy, solving $\sim$1.5$\times$ more bounds among those unsolved without inprocessing.

We have so far demonstrated BI's benefit in incremental preprocessing; however, BI enables a much broader design space. Inprocessing can now be integrated into the search in many ways, paving the way for new strategies and significantly more efficient SAT solving, both non-incremental and incremental.

{\sloppy Beyond SAT solving, BI is also promising for enabling inprocessing in state‑of‑the‑art algorithms for model counting~\cite{DBLP:conf/cav/SoosM25} and enumeration~\cite{DBLP:journals/ai/SpallittaSB25}, which rely on systematic search without restarts and have therefore so far been unable to exploit inprocessing.\par}
\bibliographystyle{plainurl} 
\bibliography{bi}         
\clearpage
\appendix
\section{Proofs}
\label{app:proofs}

This appendix contains the detailed proofs of all the lemmas deferred from the main body of the paper.

\coFromReCl*  

\begin{proof}
We prove both directions.

\smallskip\noindent
\emph{First direction:} $(N \equiv D_0 \cup D) \to^n (P \equiv D_0 \cup S^0)$.
Since $D_0 \to^n D_0$ is trivial, it suffices to show that every $0$-restored clause
$\alpha \in S_i^0$ (for some $i \le n$) is implied at level $n$ by $D_0 \cup D$.

We perform downward induction on~$i$.

\emph{Base case ($i=n$):}
By definition of $n$ (the decision level where the latest BI invocation completed),
$S_n$ is the highest stashed group, so the restored-representation case
(representation via some $\gamma \in S_q^r$ with $q>n$) cannot occur.
By Inv.~\ref{inv:re}, each $\alpha \in S_n^0$ is either trivially satisfied at level $n$
or represented by some deputy in $D$.
In both cases, $D_0 \cup D \to^n \alpha$ follows.

\emph{Inductive step ($i<n$):}
Let $\alpha \in S_i^0$.
By Inv.~\ref{inv:re}, $\alpha$ satisfies one of the following:
\begin{itemize}
  \item $\top \to^i \alpha$, whence $D_0 \cup D \to^n \alpha$ by Cor.~\ref{cor:monotone-d};
  \item $\alpha$ is $i$-implied by some deputy in $D$, whence $D_0 \cup D \to^n \alpha$ by Cor.~\ref{cor:monotone-d};
  \item $\alpha$ is implied at level $i$ by some $\gamma \in S_q^r$ with $q>i$.
        By the induction hypothesis, $\gamma$ is $j$-implied by some deputy in $D$
        for some $j \le n$. By  Cor.~\ref{cor:monotone-d}, $D_0 \cup D \to^n \alpha$.
\end{itemize}

Thus every $\alpha \in S_i^0$ is implied at level $n$ by $D_0 \cup D$,
proving $N \to^n P$.

\smallskip\noindent
\emph{Second direction:}
$P \to^n (N \equiv D_0 \cup D)$ follows from Inv.~\ref{inv:cl},
since each clause in $D_0 \cup D$ is implied by $P$ at level~0,
and hence also at level~$n$ by Cor.~\ref{cor:monotone-d}.
\end{proof}

\propreinv*

\begin{proof}
Invs.~\ref{inv:in} and~\ref{inv:cl} depend only on the contents of the
pervasive ($P$), active-non-temporary ($N$), and active ($A$) clause sets,
which are unchanged when $\gamma$ is moved to a lower group within $D$.
For Inv.~\ref{inv:re}, every stashed clause deputy-represented by
$\gamma$ before the promotion is still represented by $\gamma$ afterwards,
since $\gamma$ remains active and its rank only decreases. The
representation of all other stashed clauses is unaffected.
Finally, Inv.~\ref{inv:co} follows from Invs.~\ref{inv:re} and~\ref{inv:cl}
via Lemma~\ref{lem:co-from-re-cl}.
\end{proof}

\blesskcase*

\begin{proof}
The promotion in line~\ref{alg:subsumption:promote-k} only moves $\alpha$
to a lower group within $D$ without changing any stashed sets, so all
invariants are preserved by Lemma~\ref{lem:promotion-invariants}.

Stashing $\beta$ in $S_k^b$ in line~\ref{alg:subsumption:stash-beta} leaves
Invs.~\ref{inv:cl} and~\ref{inv:in} unaffected. We therefore focus on
Inv.~\ref{inv:re}; Inv.~\ref{inv:co} then follows immediately
from Lemma~\ref{lem:co-from-re-cl}.

By construction, $\beta$ itself becomes deputy-represented by $\alpha$.
Moreover, since $\csubsmin{k}{\alpha}{\beta}$ implies $\alpha \to^k \beta$
by Cor.~\ref{cor:csub-impl}, any stashed clause previously represented by
$\beta$ at level $k$ or higher is now deputy-represented by $\alpha$ at the
same level, whereas any stashed clause previously represented by $\beta$ at
some level $< k$ is still represented by $\beta$, but now via
restored-representation rather than deputy-representation.
\end{proof}

\bgekcase*

\begin{proof}
If $a > b$, the promotion in line~\ref{alg:subsumption:promote-b} only moves
$\alpha$ to a lower group within $D$ without changing any stashed sets, so
all invariants are preserved by Lemma~\ref{lem:promotion-invariants}.

Deleting $\beta$ in line~\ref{alg:subsumption:delete-beta} affects
Inv.~\ref{inv:in} only if $\beta$ is 0-restored (i.e., $\beta \in S^0$).
In that case, $\alpha$ must subsume $\beta$ globally (i.e.,
$\csubs{0}{\alpha}{\beta}$), so replacing $\beta$ by $\alpha$ in $D_0$ is
sound. Inv.~\ref{inv:cl} is unaffected. We therefore focus on
Inv.~\ref{inv:re}, with Inv.~\ref{inv:co} following from
Lemma~\ref{lem:co-from-re-cl}.

By assumption, $\csubsmin{k}{\alpha}{\beta}$ holds with $k \le b$, and so
$\alpha \to^b \beta$ by Cor.~\ref{cor:csub-impl} and
Cor.~\ref{cor:monotone-d}. Hence every clause that was represented via $\beta$
at some level $i \ge b$ can instead be represented via $\alpha$ at level~$i$,
so deputy-representation is preserved for all such clauses. Moreover, by
Cor.~\ref{cor:no-repr-below-rank}, no clause is represented by $\beta$ at
levels below its rank~$b$.
\end{proof}

\onestepback*

\begin{proof}
We argue invariant-by-invariant.

\emph{Inv.~\ref{inv:in} (input-equivalence).}
$P = D_0 \cup S^0$ is unchanged: deleting $D_i$ does not affect $P$, and moving
$S_i^0$ from $S^0$ to $D_0$ leaves $D_0 \cup S^0$ unchanged.

\emph{Inv.~\ref{inv:cl} (closure).}
Before the step, $P \to^0 \gamma$ for all $\gamma \in A \cup S$.
Deleting $D_i$ only removes clauses and therefore cannot break closure.
Each restored $\beta \in S_i^q$ already satisfied $P \to^0 \beta$ while stashed.
After activation, $\beta \in D_q \subseteq A$, and the same implication holds.

\emph{Inv.~\ref{inv:re} (representation).}
After the step, the remaining stashed clauses are exactly
$\bigcup_{k<i} S_k$. No $\alpha \in S_k$ with $k<i$ was represented by any clause
in $D_i$ (by Cor.~\ref{cor:no-repr-below-rank}, a rank-$i$ deputy can only represent
levels $\ge i$), so deleting $D_i$ is harmless. If $\alpha \in S_k$ was
restored-represented via some $\beta \in S_i^q$, then $\beta \to^{k} \alpha$
held before; after restoring $\beta$ into $D_q$, the same implication certifies
deputy-representation (with deputy rank $q \le k$). The remaining trivial-representation
still holds.

\emph{Inv.~\ref{inv:co} (correctness).}
Follows from Invs.~\ref{inv:re} and~\ref{inv:cl} via
Lemma~\ref{lem:co-from-re-cl}.
\end{proof}

\end{document}